\documentclass[a4paper,onecolumn,11pt]{quantumarticle}
\pdfoutput=1
\usepackage[utf8]{inputenc}
\usepackage[english]{babel}
\usepackage[T1]{fontenc}
\usepackage{csquotes}
\usepackage{amsmath}
\usepackage{amsthm}
\usepackage{mathtools}
\usepackage{amssymb}

\usepackage{hyperref}
\usepackage{diagbox}
\usepackage{booktabs}
\usepackage{tikz}
\usepackage{lipsum}
\usepackage{bm}
\usepackage{physics}
\usepackage[ruled]{algorithm2e}
\usepackage{caption}
\usepackage{adjustbox}
\usepackage{makecell}
\usepackage{placeins}

\usepackage{caption}
\usepackage{subcaption}

\usepackage[labelformat=simple]{subcaption}

\newcommand{\lmref}[1]{\hyperref[#1]{Lemma \ref*{#1}}}
\newcommand{\prref}[1]{\hyperref[#1]{Proposition \ref*{#1}}}
\newcommand{\GS}{T_{\text{g}}}

\newcommand{\GW}{Goemans-Williamson\ }
\newcommand{\mdc}{:}

\newcommand{\fref}[1]{\autoref{#1}}

\newcommand{\sref}[1]{\autoref{#1}}
\newcommand{\Sref}[1]{\autoref{#1}}
\newcommand{\tref}[1]{\autoref{#1}}

\newcommand{\aref}[1]{\autoref{#1}}

\usepackage{mleftright}
\mleftright

\usetikzlibrary{quantikz}

\renewcommand{\cross}{\textsc{cut}_{01}}
\newcommand{\inner}[2]{\textsc{cut}_{#1}\left(#2\right)}
\newcommand{\best}{\textsc{opt-rem}}
\newcommand{\cxi}{\textsc{rem}}

\newcommand{\poly}{\mathrm{poly}}
\newcommand{\cover}{\textsc{cover}}
\newcommand{\zo}[1]{\{0,1\}^{#1}}

\newcommand{\mxct}{\textsc{Max-Cut}}
\newcommand{\sbest}{\textsc{opt-rem'}}

\DeclareMathOperator{\meop}{e}
\newcommand{\me}[1]{\meop\limits^{#1}}

\newcommand{\sfrac}[2]{#1/#2}

\newcommand{\mi}{\mathrm{i}}
\newcommand{\rhox}[1]{\Gamma_{x,#1}}
\newcommand{\etal}{\textit{et al. }}

\newcommand{\tian}[1]{\textcolor{red}{#1}}

\newcommand{\added}[1]{\textcolor{purple}{#1}}

\DeclareMathOperator*{\argmax}{argmax}

\SetKwComment{Comment}{/* }{ */}

\newcommand{\ER}{Erdős-Rényi }

\newcommand{\bgb}{\bm{\beta},\bm{\gamma}}

\newtheorem{prop}{Proposition}
\newtheorem{lemma}{Lemma}

\newcommand{\cmt}[1]{}

\newcommand{\cxiss}{\cxi\left(s_0,s_1\right)}

\newenvironment{npenv}[1][\unskip]{%
\par\vspace{0.3cm}
\noindent
\textbf{#1}
\noindent \it}
{\vspace{0.3cm}}

\usepackage[backend=bibtex,sorting=none]{biblatex}

\addto\extrasenglish{

}

\begin{document}




\title[QAOA with fewer qubits: a coupling framework]{QAOA with fewer qubits: a coupling framework to solve larger-scale Max-Cut problem}

\author[1,2]{Yiren Lu}
\author[1,2]{Guojing Tian}
\author[1,2]{Xiaoming Sun}

\address{State Key Lab of Processors, Institute of Computing Technology, CAS, Beijing 100190, China}
\address{School of Computer Science and Technology, University of Chinese Academy of Sciences, Beijing 100049, China}


\begin{abstract}
  Maximum cut (Max-Cut) problem is one of the most important combinatorial optimization problems because of its various applications in real life, and recently Quantum Approximate Optimization Algorithm (QAOA) has been widely employed to solve it. However, as the size of the problem increases, the number of qubits required will become larger. With the aim of saving qubits, we propose a coupling framework for designing QAOA circuits to solve larger-scale Max-Cut problem.
  This framework relies on a classical algorithm that approximately solves a certain variant of Max-Cut, and we derive an approximation guarantee theoretically, assuming the approximation ratio of the classical algorithm and QAOA. Furthermore we design a heuristic approach that fits in our framework and perform sufficient numerical experiments, where we solve Max-Cut on various $24$-vertex \ER graphs. Our framework only consumes $18$ qubits
  and achieves $0.9950$ approximation ratio on average, which outperforms the previous methods showing $0.9778$ (quantum algorithm using the same number of qubits) and $0.9643$ (classical algorithm). The experimental results indicate our well-designed quantum-classical coupling framework gives satisfactory approximation ratio while reduces the qubit cost, which sheds light on more potential computing power of NISQ devices.
\end{abstract}

\section{Introduction}

Quantum Approximate Optimization Algorithm (QAOA) is a hybrid quantum-classical approach for combinatorial optimization problem, proposed by Farhi \etal  \cite{farhi_quantum_2014}. Here, \textit{hybrid} means that QAOA composes of a parameterized circuit (ansatz) whose measurement outcome can be mapped to the value of the objective function, and the parameters are updated with a classical optimizer that maximizes (or minimizes) the expected value of the objective function. Delegating some computational processes to the classical domain grants QAOA a comparatively shallower circuit, making it more applicable to NISQ \cite{preskill_quantum_2018} devices. Also, as the circuit for QAOA is a parameterized one, it is easier to be modified and adapted to different types of quantum computers \cite{headley_approximating_2022}. Moreover, the quantum gates required to implement a QAOA circuit are comparatively simple. For instance, the QAOA circuit for solving the maximum cut (Max-Cut) problem only comprises CNOT and single qubit rotation gates, which are much more friendly to physical realization than Toffoli gates, etc. Therefore, QAOA is more compatible with NISQ devices than other complex quantum algorithms such as Shor's algorithm \cite{shor_polynomial-time_1997}, Grover's algorithm \cite{grover_fast_1996}.


Since its introduction, QAOA has attracted considerable theoretical and experimental attention \cite{crooks_performance_2018,farhi_quantum_2022,basso_quantum_2022,harrigan_quantum_2021,wurtz_maxcut_2021,sack_quantum_2021,zhou_quantum_2020}. Starting from the original paper of QAOA  \cite{farhi_quantum_2014}, Max-Cut problem has always been used as a good example to demonstrate QAOA's capabilities. Max-Cut problem is a classical NP-hard problem in graph theory and has its applications in statistical physics and integrated circuit design. Furthermore, as a combinatorial optimization problem with simple constraints and objective function, solutions to Max-Cut problem can be potentially generalized to other hard optimization problems \cite{waldspurger_phase_2015}. This paper also investigates how to solve the Max-Cut problem with QAOA.

When using QAOA to solve Max-Cut problem on a graph, the number of qubits required is usually equal to the number of vertices in this graph.
However, the scaling of qubit numbers in quantum computer physical realization proves to be a challenging task, and the number of available qubits will likely be constrained in the near future. Thus, it remains a good question whether we can reduce the number of qubits required to deal with a specific  problem instance.

Efforts have been made in \cite{li_large-scale_2023, zhou_qaoa--qaoa_2023}, where the authors split the graph into several subgraphs and solve Max-Cut on the subgraphs independently, and then combine the cuts of subgraphs into a global cut. In the combination step, both methods only keep a \textit{small} amount of candidate solutions, determined by their optimality inside the subgraph, and then try to select an optimal combination of the kept candidate solutions in subgraphs to maximize the overall cut. These methods allow for using an arbitrary small number of qubits for a given problem instance, by recursively partitioning the graph until the number of vertices is smaller than the number of available qubits. Moreover, for graphs that innately compose of several dense subgraphs connected by a few edges, such splitting and combining strategy is reasonable for the cut is largely determined by edges inside the dense subgraphs, and the edges joining the subgraphs are relatively insignificant. But for general cases, Max-Cut is not a problem with much locality — a cut may not be so good in a subgraph, but can be a part of a nice solution for the whole graph. Therefore, due to the existence of such bold greedy reduction that discards locally-bad solutions without much analysis, it is hard to mathematically derive a strong performance guarantee for such methods. 

Considering the issues of previous works mentioned above, we propose a method that incorporates a classical approximation algorithm as part of the QAOA circuit, leaving some vertices and edges to be handled by classical computation. By combining quantum and classical computation, our method applies to more general graph instances and achieves better experimental results than previous works that only used quantum computation. In addition to numerical results, we also derive a formula that lower bounds the performance guarantee with some mathematical analysis, assuming the approximation ratio of the classical algorithm and the quantum approximate optimization. Moreover, as we want to keep the overall method at least better than the classical one, we also analyse the portion of qubits that can be saved while this condition holds.


This paper is organized as follows. \Sref{sec:prelim} reviews the definition and computational hardness of Max-Cut, along with the overall scheme of QAOA and the QAOA circuit for Max-Cut. \Sref{sec:circdes} presents a general framework of our method and analyses its performance
. \Sref{sec:simu} discusses the implementation of the classical algorithm required in this framework and gives the simulation results compared with previous works, and \sref{sec:conc} concludes with a summary of our results and some further discussions.

\section{Preliminaries}
\label{sec:prelim}

This section reviews the definition of the maximum cut (Max-Cut) problem and the structure of QAOA, which are fundamental to the following parts.

\subsection{Max-Cut problem}

As a well-known NP-hard problem, Max-Cut problem is described as follows. 

\begin{npenv}[Max-Cut.] 
  For a graph $G=(V,E)$, find a partition $(S,V-S)$ for $V$ such that
  \begin{equation}
    \# \left\{(x,y)\mdc(x,y)\in E\wedge (x\in S)\wedge (y\in V-S)\right\}
  \end{equation}
  is maximized.
  Or equivalently, assign two colors to the vertices so that
  \begin{equation}
    \# \left\{(x,y)\mdc(x,y)\in E\wedge \text{vertex $x$ and vertex $y$ have different colors}\right\}
  \end{equation}
  is maximized, where $\#S$ denotes the cardinality of set $S$.
\end{npenv}

Currently, there is no polynomial time algorithm that can solve Max-Cut accurately, so efficient algorithms proposed to solve Max-Cut problem only give approximate solutions \cite{goemans_improved_1995,karloff_how_1996,  festa_randomized_2002, khot_optimal_2007, mathieu_yet_2008}. Usually, these approaches have been benchmarked by the approximation ratio of obtained solutions, which is the value of obtained solution divided by the optimal solution.
In classical computation, the approximation ratio of any efficient algorithm is upper bounded under currently unfalsified complexity conjectures, see \tref{CLSUB}.
\begin{table}[h]
  \centering
  \begin{tabular}{ll}
    \toprule
    $\alpha_{opt}$                                                  & Assumption \\
    \midrule
    \makecell[l]{$\approx .878$ \cite{khot_optimal_2007}                       \\ (\GW algorithm)}  &  \makecell[l]{$P\not=NP$\\ \textit{unique games conjecture}}\\
    \midrule
    $\approx .941$ \cite{hastad_optimal_2001,trevisan_gadgets_1996} & $P\not=NP$ \\
    \bottomrule
  \end{tabular}

  \caption{Approximation ratio upper bound, denoted by $\alpha_{opt}$, for classical algorithms on Max-Cut.}
  \label{CLSUB}
\end{table}


But it is worth mentioning that the upper bounds apply to the worst-case approximation ratio on any graph. In other words, for any efficient classical algorithm, there will be an infinite set of extreme cases where it cannot do better than $\alpha_{\mathrm{opt}}$. But on average cases, classical algorithms usually give much better results (with approximation ratio higher than $.878$). For example, the best-known algorithm for general cases, \GW algorithm, achieves no less than $.878$ on any graph, but on \ER random graphs, it usually shows over $.9$-approximation in experiments. So the reader might not be confused when seeing such cases, and when comparing a quantum algorithm with classical ones experimentally, we need to do case-by-case comparison, as indicated in \cite{crooks_performance_2018}. 

Now we describe two classical algorithms which will be mentioned in following sections.
\paragraph{Naive random algorithm}
\label{naivealgo}
For a graph $G$, we independently assign each vertex a random color. This results a random algorithm with $\frac 12$ approximation ratio, namely the expected output is at least half of the optimal solution.
\begin{algorithm}[ht]
  \caption{Naive Max-Cut algorithm}
  \label{MC2APPROX}
  \KwData {A graph $G=(V,E)$.}
  \KwResult {A cut that achieves expected $\frac 12$-approximation.}
  $S \gets \emptyset$;

  \For{$u \in V$}{
    $u$ is added to $S$ with $\frac 12$ probability;
  }
  \Return $(S,V-S)$;
\end{algorithm}









\paragraph{Local search} A local search algorithm starts from an
arbitrary feasible solution to the problem, and incrementally modifies the solution until reaching a local optimum. 
For the case of Max-Cut, a local search algorithm might look like \aref{LSALGO}. In the pseudo code, $S\vartriangle \{u\}$ means symmetric difference of $S$ and $\{u\}$, i.e., removing $u$ from $S$ if $u\in S$ or add $u$ to $S$ if $u\not\in S$. The procedure is bound to stop in polynomial iterations, and guarantees to give a solution with at least $\frac 12$ approximation ratio.
\begin{algorithm}[ht]
  \caption{Local search for Max-Cut}
  \label{LSALGO}
  \KwData {A graph $G=(V,E)$.}
  \KwResult {A cut that achieves $\frac 12$-approximation.}
  $S\gets S_0$\Comment*[r]{some initial partition solution}

  \While{$\exists u, \textsc{cut}(S\vartriangle \{u\}) > \textsc{cut}(S)$}{
    $S\gets S\vartriangle \{u\}$; 
  }
  \Return $(S,V-S)$;
\end{algorithm}

Though local search does not theoretically ensure a high approximation ratio, it does demonstrates good performance in experiments. Moreover, it is quite simple and scalable, so in later sections we implement our heuristic method on top of it.

\subsection{QAOA}

Besides the classical algorithms for Max-Cut problem, some QAOA-based quantum algorithms have been introduced recently. QAOA is a quantum-classical hybrid method for combinatorial optimization, solving problems that look as follows \cite{farhi_quantum_2014}:
\begin{center}
  Given $C:\zo{n}\mapsto \mathbb{R}$, find $\displaystyle \max_{s\in\{0,1\}^n} C(s)$.
\end{center}
This cost function $C$ can be encoded into a problem Hamiltonian \cmt{\tian{the scope of $s$}}
\begin{equation}
  H_C = \sum_{s\in \zo{n}} C(s)\ket{s}\bra{s},
\end{equation}
and in QAOA, the combinatorial optimization problems are reformulated into approximately finding the \cmt{\tian}{ground state and} ground energy of $-H_C$, which is done with a hybrid optimization procedure. Firstly, prepare the parameterized unitary transformation, which can generate the following state when operating on $\ket{+}^{\otimes n}$, i.e.,
\begin{equation}
  \ket{\bm{\beta},\bm{\gamma}}= \left(\prod_{k=1}^p \exp(-\mi \beta_{k} H_B) \exp(-\mi \gamma_{k} H_C) \right)\ket{+}^{\otimes n},
  \label{QAOACIR}
\end{equation}
with $H_B=\sum_{i=1}^n X_i$, $p$ being the number of layers for QAOA and $\bm{\beta, \gamma}\in [0,2\pi)^p$ being parameter vectors.
Secondly, update the parameters $\bm{\beta},\bm{\gamma}$ using a classical optimizer to maximize the expected measurement outcome,
\begin{equation}
  (\bm{\beta}^*,\bm{\gamma}^*)= \argmax\limits_{\bgb} \bra{\bm{\beta},\bm{\gamma}}H_C\ket{\bm{\beta},\bm{\gamma}}.
\end{equation}
Thus, the approximation ratio will be
\begin{equation}
  \frac {\bra{\bm{\beta}^*,\bm{\gamma}^*}H_C\ket{\bm{\beta}^*,\bm{\gamma}^*}}{\displaystyle \max_s C(s)}.
\end{equation}
\cmt{\tian}{The author of \cite{farhi_quantum_2014} have proved that the approximation ratio approaches $1$ when $p$ tends to infinity.}

Before explaining how to realize the above process in the language of quantum circuit, it is necessary to review some elementary quantum gates we will use in our paper, including 1-qubit gates $R_x$ (rotation along $X$-axis), $R_z$ (rotation along $Z$-axis), $H$ (Hadamard gate),
$$
  R_x(\theta)=\begin{pmatrix}
    \cos(\theta/2)     & -\mi\sin(\theta/2) \\
    -\mi\sin(\theta/2) & \cos(\theta/2)
  \end{pmatrix},
  R_z(\theta)=\begin{pmatrix}
    \me{-\mi(\theta/2)} & 0                  \\
    0                   & \me{\mi(\theta/2)}
  \end{pmatrix},
  H=\frac{1}{\sqrt{2}}\begin{pmatrix}
    1 & 1  \\
    1 & -1
  \end{pmatrix},
$$
(where $\theta\in\mathbb{R}$ is a parameter denoting rotation angle) and 2-qubit controlled-NOT gate along with its circuit notation
\begin{figure}[h]
  \centering
  \begin{subfigure}[b]{0.6\textwidth}
    $$
      \text{CNOT}=\begin{pmatrix}
        1 & 0 & 0 & 0 \\
        0 & 1 & 0 & 0 \\
        0 & 0 & 0 & 1 \\
        0 & 0 & 1 & 0 \\
      \end{pmatrix}
    $$
  \end{subfigure}
  \begin{subfigure}[b]{0.4\textwidth}
    \begin{quantikz}
      \lstick{$\ket{a}$}&\ctrl{1}&\rstick{$\ket{a}$}\qw\\
      \lstick{$\ket{b}$}&\targ{}&\rstick{$\ket{b \oplus a}$}\qw
    \end{quantikz},
  \end{subfigure}
\end{figure}
\\which means the target qubit (marked with $\bigoplus$) is flipped iff the control qubit (marked with $\bullet$) is in state $\ket{1}$.
\setcounter{figure}{0}

Now we show how to construct the unitaries above required by QAOA. The $\exp\left(-\mi\beta_k H_B\right)$ part is straightforward as it is simply a layer of $R_x(\beta_k)$ gates.
Now we consider $\exp\left(-\mi\gamma_k H_C\right)$. When solving Max-Cut problem with QAOA, we use an $s\in \zo{n}$ to encode the coloring of vertices, and $C(s)$ would be the number of edges whose endpoints are colored differently. And for such $C(s)$, $\exp(-\mi \gamma_{k} H_C)$ can be implemented by applying the circuit shown in \fref{UXY} for all edges $(x,y)$. Apparently, the circuit in \fref{UXY} add phase factors
depending on whether $x$-th and the $y$-th qubits are in different states, and the phase factors thus accumulate to the whole cut size (up to a global phase factor).

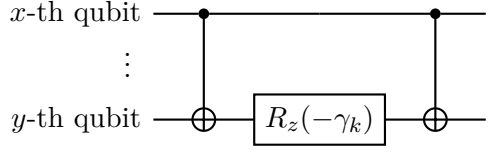
\begin{figure}[t]
  \centering
  \def\myldots{\ \ldots\ \qw}
  \def\myvdots{\ \vdots\ }
  \adjustbox{scale=1,center}{
    \begin{quantikz}
      \lstick{$x$-th qubit}&\ctrl{2}&\qw &\ctrl{2}& \qw\\
      \lstick{\myvdots} & \\
      \lstick{$y$-th qubit}&\targ{} &\gate{R_z(-\gamma_k )}&\targ{} & \qw
    \end{quantikz}
  }
  \caption{Adding a phase factor depending on whether two qubits $x,y$ shares the same value.}
  \label{UXY}
\end{figure}



As QAOA inherits the spirit of quantum adiabatic evolution, the approximation ratio gradually increases to $1$ when the hyperparameter $p$ increases to infinity. But with moderate circuit depth, QAOA hasn't been proved theoretically to achieve a better approximation ratio than best-known classical algorithms. However, in numerical experiments, QAOA does seem to show nice performance in shallow circuit depth. For example, Crooks \cite{crooks_performance_2018} shows for \ER graphs with up to $17$ vertices, QAOA significantly outperforms \GW algorithm when the hyperparameter $p$ grows to $8$, and the approximation ratio is very close to $1$ when $p$ is set $32$. Therefore, it is promising to achieve quantum advantage with QAOA, namely, obtain better cut solutions on certain problem instances than efficient classical algorithms. 

In the following sections, we describe our coupling framework which embeds classical approximation algorithm with QAOA which will be able to use fewer qubits to solve Max-Cut problem. We demonstrate its performance theoretically and numerically. Due to the good numerical performance of QAOA, our framework indeed gives good results in numerical experiments and outperforms currently the best-known classical algorithm.




\section{Theoretical results}
\label{sec:circdes}

Despite the proof of quantum advantage for the original QAOA referred in \cite{farhi_quantum_2014} being hard, for a new coupling framework below we are able to give some theoretical results assuming the approximation ratio of QAOA is known.
Firstly we give a reformulation of Max-Cut that leads to our new framework. Then we present details of our framework and give some performance analysis.

\subsection{Framework description} 

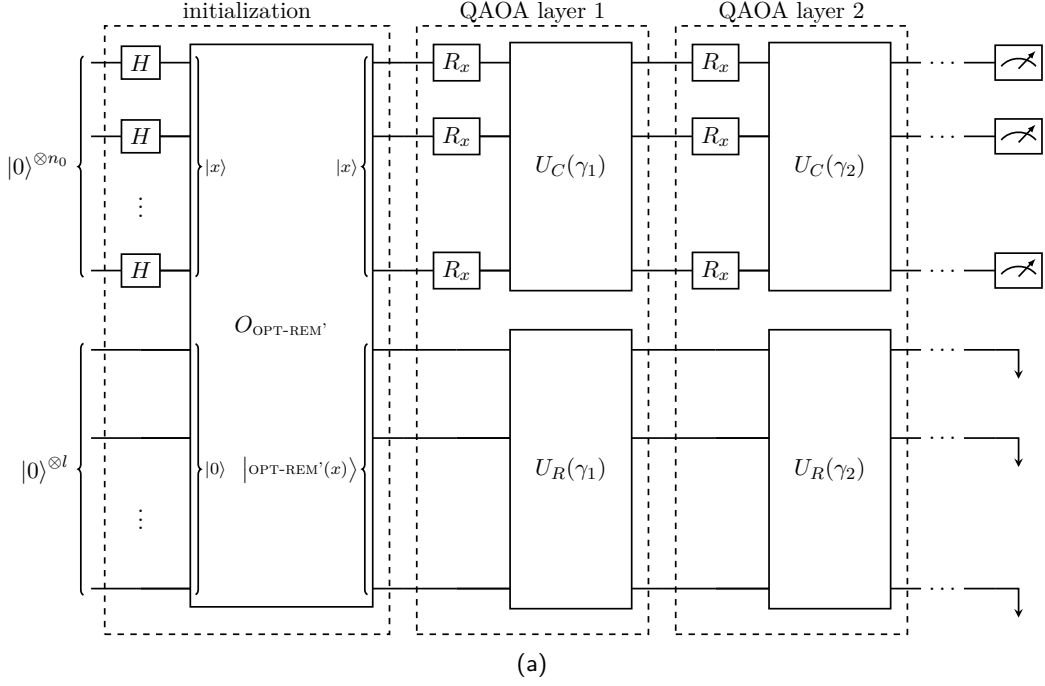
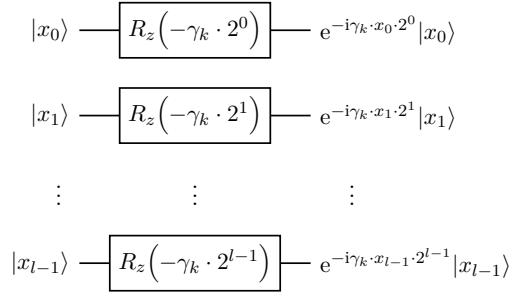
\begin{figure}[!ht]
  \begin{subfigure}[b]{\textwidth}
    \centering
    \def\qwb{\qwbundle[alternate]{}}
    \def\myldots{\ \ldots\ \qw}
    \def\myvdots{\ \vdots\ }
    \adjustbox{scale=0.8,center}{
      \begin{quantikz}
        \lstick[wires=4]{$\ket{0}^{\otimes n_0}$}& \gate{H}\gategroup[wires=8,steps=2,style=dashed]{initialization}& \gate[8,nwires={3,7}][3cm]{O_{\sbest}}\gateinput[4]{$\ket{x}$}\gateoutput[wires=4]{$\ket{x}$}& \qw& \gate{R_x}\gategroup[wires=8,steps=2,style=dashed]{QAOA layer 1}& \gate[4,nwires={3}][2cm]{U_C(\gamma_1)}& \qw& \gate{R_x}\gategroup[wires=8,steps=2,style=dashed]{QAOA layer 2}& \gate[4,nwires={3}][2cm]{U_C(\gamma_2)}& \myldots& \meter{}\\
        \lstick{}& \gate{H}& \qw& \qw& \gate{R_x}& \qw& \qw& \gate{R_x}& \qw& \myldots& \meter{}\\
        \lstick{}& \myvdots& & & & & & & & & & \\
        \lstick{}& \gate{H}& \qw& \qw& \gate{R_x}& \qw& \qw& \gate{R_x}& \qw& \myldots& \meter{}\\
        \lstick[wires=4]{$\ket{0}^{\otimes l}$}& \qw& \qw\gateinput[4]{$\ket{0}$}\gateoutput[wires=4]{$\ket{\sbest(x)}$}& \qw& \qw& \gate[4,nwires={3}][2cm]{U_R(\gamma_1)}& \qw& \qw& \gate[4,nwires={3}][2cm]{U_R(\gamma_2)}& \myldots& \trash{}\\
        \lstick{}& \qw& \qw& \qw& \qw& \qw& \qw& \qw& \qw& \myldots& \trash{}\\
        \lstick{}& \myvdots& & & & & & & & & & \\
        \lstick{}& \qw& \qw& \qw& \qw& \qw& \qw& \qw& \qw& \myldots& \trash{}
      \end{quantikz}
    }
    \caption{}
    \label{ORCCIR}
    \vspace{1em}
  \end{subfigure}

  \begin{subfigure}[b]{\textwidth}
    \centering
    \def\qwb{\qwbundle[alternate]{}}
    \def\myldots{\ \ldots\ \qw}
    \def\myvdots{\ \vdots\ }
    \adjustbox{scale=0.8,center}{
      \begin{tikzcd}
        \lstick{$\ket{x_{0}}$} & \gate{R_z\left(-\gamma_k\cdot 2^0\right)}&\rstick{$\me{-\mi\gamma_k\cdot x_{0}\cdot 2^{0}}\ket{x_{0}}$}\qw \\
        \lstick{$\ket{x_{1}}$} & \gate{R_z\left(-\gamma_k\cdot 2^1\right)}&\rstick{$\me{-\mi\gamma_k\cdot x_{1}\cdot 2^{1}}\ket{x_{1}}$}\qw \\
        \lstick{\myvdots} & \myvdots && \myvdots\\
        \lstick{$\ket{x_{l-1}}$} & \gate{R_z\left(-\gamma_k\cdot 2^{l-1}\right)}&\rstick{$\me{-\mi\gamma_k\cdot x_{l-1}\cdot 2^{l-1}}\ket{x_{l-1}}$}\qw \\
      \end{tikzcd}
    }
    \caption{}
    \label{TOPHASE}
  \end{subfigure}
  \caption{(a) The circuit implementation under general settings, assuming an access to an oracle $O_{\sbest}$. (b) The implementation detail for $U_R$ gate in (a).}
\end{figure}

When solving the Max-Cut problem on a graph, clearly a denser subgraph tends to have a larger impact on the overall solution. Suppose we select a dense subgraph $G_0=(V_0,E_0)$, and define $V_1=V-V_0, G_1=(V_1,E_1), n_0=|V_0|, n_1=|V_1|$.


Once we fix a coloring $s_0\in \zo{n_0}$ of $V_0$, the cut size in $E_0$ is fixed, and we denote it with $\inner{0}{s_0}$. Now, choosing different coloring $s_1\in \zo{n_1}$ of $V_1$ will result different cut sizes in the remaining edges of the graph (or formally, $E-E_0$), defined as
\begin{equation}
  \cxiss =\cross(s_0,s_1)+\inner{1}{s_1},
\end{equation}
where $\cross(s_0,s_1)$ stands for the number of cut edges in $E-E_0-E_1$, and $\inner{1}{s_1}$ denotes the number of cut edges in $E_1$. For each $s_0$, there will be an optimal $s_1$ that maximizes the size of cut in the remaining edges of the graph, and we define this maximum as
\begin{equation}
  \best(s_0)    = \max_{s_1 \in \{0,1\}^{n_1}}\cxiss. \label{OMGDEF}
\end{equation}
With the definitions above, Max-Cut problem is reformulated into
\begin{equation}
  \mxct(G)=\max_{s_0\in\{0,1\}^{n_0}}\left\{\inner{0}{s_0}+\best(s_0)\right\}.
  \label{REFCUT}
\end{equation}
Next, we will use QAOA to solve the optimization problem \eqref{REFCUT}, with the help of a classical algorithm that approximately computes $\best(s_0)$.

Remember in QAOA, to optimize $\max_{s\in \zo{n}} C(s)$, we need to construct the problem unitary, which implements
\begin{equation}
  \ket{s}\mapsto \me{-\mi\gamma_{k}\cdot C(s)}\ket{s}.
\end{equation}
Accordingly, here, a problem unitary for \eqref{REFCUT} would have to contain \cmt{\tian{in???}}two steps:
\begin{align}
  \ket{s_0}\xmapsto{\text{step 1}} & \me{-\mi\gamma_{k}\cdot \inner{0}{s_0}}\ket{s_0}                                            \\
  \xmapsto{\text{step 2}}          & \me{-\mi\gamma_{k}\cdot  \best(s_0)}\cdot \me{-\mi\gamma_{k}\cdot \inner{0}{s_0}}\ket{s_0}.
  \label{SHFT}
\end{align}

\paragraph{Step 1}
The unitary $\ket{s_0}\mapsto \me{-\mi\gamma_{k}\cdot \inner{0}{s_0}}\ket{s_0}$ is the same as that in the original QAOA circuit for Max-Cut , and can be implemented by applying the circuit \fref{UXY} for all $(x,y) \in E_0$.

\paragraph{Step 2}
First let us consider the general case, where we will implement the classical algorithm as an oracle shown below (more discussions on this classical oracle is given in \sref{sec:simu}),
\begin{equation}
  O_{\sbest}:\ket{s_0}\ket{0}^{\otimes l}\mapsto \ket{s_0}\ket{\sbest(s_0)}.
  \label{CLSORAC}
\end{equation}
Here, we need to store the largest possible $\sbest(s_0)$ on the second $l$-qubit register, thus $l$ should be set $\left\lceil\log \left(|E|-|E_0|\right)\right\rceil \in O\left(\log n\right)$.

With the oracle defined in \eqref{CLSORAC}, we can implement the required phase oracle
$$
  \ket{s_0}\ket{\sbest(s_0)}\mapsto \me{-\mi\gamma_{k}\cdot  \best(s_0)}\ket{s_0}\ket{\sbest(s_0)},
$$
using only rotation gates, as shown in \fref{TOPHASE}. And the overall circuit requires $n_0+\left\lceil\log \left(|E|-|E_0|\right)\right\rceil = n_0+O(\log n)$ qubits, shown in \fref{ORCCIR}.

\subsection{Performance analysis}
Next we will analyse the performance of the above framework in terms of approximation ratio and the size of the chosen subgraph (which determines the number of saved qubits).

Suppose the classical oracle achieves $\beta$-approximation ratio $(0<\beta \le 1)$, that is,
\begin{equation}
  \frac {\sbest(s)}{\best(s)}=\beta+\varepsilon_s,
\end{equation}
where $\varepsilon_s\ge 0$ is a slack variable. As the classical oracle results approximation loss, the QAOA procedure in our framework might not be able to discover the optimal solution, as it now tries to solve the following optimization problem:

\cmt{\tian{maybe in the form of lemma or proposition? and the following is the detailed proof.}}
\begin{equation}
  \label{newopt}\max_s\{a_s+(\beta+\varepsilon_s) b_s\},
\end{equation}
where $a_s=\inner{0}{s}, b_s=\best(s)$. Suppose the QAOA part solves \eqref{newopt} in $\alpha$-approximation ratio $(0 < \alpha \le 1)$ and produces expected value $\mathbb{E}(\text{QAOA})$, which means
\begin{equation}
  \frac {\mathbb{E}(\text{QAOA})}{\max_s\{a_s+(\beta+\varepsilon_s) b_s\}}\ge \alpha,
\end{equation}
then we have
\begin{lemma}[Approximation ratio lower bound]
  When we use coupling framework with $\alpha$-approx QAOA and $\beta$-approx classical algorithm, the approximation ratio for the overall method $\gamma$ satisfies
  \begin{equation}
    \gamma \ge \alpha\cdot\left(1+(\beta-1)\frac {\max_s\{b_s\}}{\max_s\{a_s+b_s\}}\right),
  \end{equation}
  \cmt\tian{where $a_s=\inner{0}{s}, b_s=\best(s)$.}
\end{lemma}
\begin{proof}
  The proof mainly consists of simple inequalities.
  \begin{align}
    \gamma= & \frac {\mathbb{E}(\text{QAOA})}{\max_s\{a_s+b_s\}}                                                                                           \\
    =       & \frac {\mathbb{E}(\text{QAOA})}{\max\{a_s+(\beta+\varepsilon_s) b_s\}}\cdot \frac {\max\{a_s+(\beta+\varepsilon_s) b_s\}}{\max_s\{a_s+b_s\}} \\
    \ge     & \alpha \cdot \frac {\max\{a_s+(\beta+\varepsilon_s) b_s\}}{\max_s\{a_s+b_s\}}
    \\
    \ge     & \alpha\cdot \frac {\max_s\{a_s+\beta b_s\}}{\max_s\{a_s+b_s\}}                                                                               \\
    \ge     & \alpha\cdot\frac {\max_s\{a_s+b_s\}+\min_s \{(\beta-1)b_s\}}{\max_s\{a_s+b_s\}} \label{MINSTEP}                                              \\
    =       & \alpha\cdot\left(1+(\beta-1)\frac {\max_s\{b_s\}}{\max_s\{a_s+b_s\}}\right), \label{GAMMAB}
  \end{align}
  where \eqref{MINSTEP} is due to $\beta-1$ being negative so the $\max$ operation flips to $\min$.
\end{proof}

As the framework embeds a $\beta$-approximation algorithm as a subprogram of an $\alpha$-approximation algorithm, the overall approximation ratio would be lower than $\alpha$ (otherwise we magically save the number of qubits with no cost). But as a basic requirement, we need to ensure the overall framework achieves better performance than the classical algorithm (otherwise we can just use the classical algorithm to solve Max-Cut and dump the quantum part). The following lemma gives a sufficient condition to ensure that requirement, namely, $\gamma>\beta$.

\begin{lemma}[Sufficient condition for quantum advantage]
  \label{LEMMA2}
  To have $\gamma>\beta$, it suffices to ensure that $n_0$ is large enough, namely
  \begin{equation}
    \frac {{n_0 (n_0-1)}}{{n (n-1)}}> 1- \frac 12\left(\frac 1{\alpha}+\frac 1{1-\beta}-\frac {1}{\alpha(1-\beta)}\right).
    \label{FINALBOUND}
  \end{equation}
\end{lemma}
To prove the lemma above, we need the following proposition.
\begin{prop}[Number of edges in a random subgraph]
  \label{PROP1}

  For a graph $G=(V,E)$, when we select a random subgraph $G_0=(V_0,E_0)$ with $|V_0|=n_0$, we have
  \begin{equation}
    \mathbb{E}(|E_0|)=|E|\cdot \frac {{n_0\cdot (n_0-1)}}{{n\cdot (n-1)}}.
  \end{equation}
\end{prop}

\begin{proof}
  There are $\binom{n}{n_0}$ subgraphs sized $n_0$.

  For each edge, it is included in a subgraph if and only if the subgraph contains its two endpoints, so it is included in $\binom{n-2}{n_0-2}$ subgraphs. So the total number of edges in the $\binom{n}{n_0}$ subgraphs mentioned above is
  \begin{equation}
    |E|\binom{n-2}{n_0-2},
  \end{equation}
  giving
  \begin{equation}
    \mathbb{E}(|E_0|)=\frac {|E|\binom{n-2}{n_0-2}}{\binom{n}{n_0}}=|E|\cdot \frac {n_0(n_0-1)}{n(n-1)}.
  \end{equation}
\end{proof}

Now we give the proof of \lmref{LEMMA2}.

\begin{proof}\textit{(of \lmref{LEMMA2})}
  To have the right hand side of \eqref{GAMMAB} no less than $\beta$, we need
  $$
    \frac {\max_s\{b_s\}}{\max_s\{a_s+b_s\}} < \frac {\alpha-\beta}{\alpha(1-\beta)}=\frac 1{\alpha}+\frac 1{1-\beta}-\frac {1}{\alpha(1-\beta)}.
  $$
  Meanwhile, we have
  \begin{align}
                & \max_s\{b_s\} \le  |E-E_0|,                                                   \\
                & \max_s\{a_s+b_s\}=\mxct(G)\ge \frac {|E|}2,                                   \\
    \Rightarrow & \frac {\max_s\{b_s\} }{\max_s\{a_s+b_s\}}\le \frac {|E-E_0|}{\frac {|E|}{2}},
  \end{align}
  so to ensure $\gamma>\beta$, it suffices to satisfy
  \begin{equation}
    \begin{aligned}
       & \frac {|E-E_0|}{|E|} < \frac 12\left(\frac 1{\alpha}+\frac 1{1-\beta}-\frac {1}{\alpha(1-\beta)}\right), \\
    \end{aligned}\label{RATINEQ}
  \end{equation}
  that is,
  \begin{equation}
    \frac {|E_0|}{|E|} >1- \frac 12\left(\frac 1{\alpha}+\frac 1{1-\beta}-\frac {1}{\alpha(1-\beta)}\right).
  \end{equation}


  From \prref{PROP1}, we can obtain a lower bound for $\sfrac {|E_0|}{|E|}$ that is related to $n_0$, because since we have the expected number of edges for a random subgraph with $n_0$ vertices, we can assert there are comparatively denser subgraphs with at least $\mathbb{E}(|E_0|)$ edges. We can find such a subgraph easily
  , then we have
  \begin{equation}
    |E_0|\ge |E|\cdot \frac {n_0(n_0-1)}{n(n-1)} \Rightarrow \frac {|E_0|}{|E|}\ge \frac {n_0(n_0-1)}{n(n-1)}.
    \label{EBAVG}
  \end{equation}
  Thus, to \textbf{keep $\gamma > \beta$ to remain advantage over the classical algorithm}, it suffices to ensure
  \begin{equation}
    \frac {{n_0 (n_0-1)}}{{n (n-1)}}> 1- \frac 12\left(\frac 1{\alpha}+\frac 1{1-\beta}-\frac {1}{\alpha(1-\beta)}\right).
  \end{equation}
\end{proof}

\begin{table}[ht]
  \centering

  \begin{subtable}[h]{\textwidth}
    \footnotesize
    \centering
    \begin{tabular}{c|lllllllll}
      \toprule
      \diagbox[]{$\alpha$}{$\beta$} & $0.85$ & $0.87$ & $0.89$ & $0.91$ & $0.93$ & $0.95$ & $0.97$ & $0.99$ \\ \midrule
      $0.85$                        & $1.0$                                                                 \\
      $0.87$                        & $0.92$ & $1.0$                                                        \\
      $0.89$                        & $0.85$ & $0.91$ & $1.0$                                               \\
      $0.91$                        & $0.78$ & $0.83$ & $0.9$  & $1.0$                                      \\
      $0.93$                        & $0.71$ & $0.75$ & $0.8$  & $0.88$ & $1.0$                             \\
      $0.95$                        & $0.65$ & $0.68$ & $0.71$ & $0.77$ & $0.85$ & $1.0$                    \\
      $0.97$                        & $0.59$ & $0.6$  & $0.63$ & $0.66$ & $0.71$ & $0.79$ & $1.0$           \\
      $0.99$                        & $0.53$ & $0.53$ & $0.54$ & $0.55$ & $0.57$ & $0.6$  & $0.66$ & $1.0$  \\ \bottomrule
    \end{tabular}
    \caption{}
    \label{VTABLE}
    \vspace{1em}
  \end{subtable}
  \begin{subtable}[h]{\textwidth}
    \footnotesize
    \centering
    \begin{tabular}{l|r|rr|rr}
      \toprule
      Approx      &       & \multicolumn{2}{l|}{$\alpha=0.97, \beta=0.91$} & \multicolumn{2}{l}{$\alpha=0.95, \beta=0.89$}                      \\ \midrule
      Device      & $n_0$ & max $n$                                        & $\Delta$                                      & max $n$ & $\Delta$ \\  \midrule
      Tianmu-1    & 36    & 44                                             & 8                                             & 42      & 6        \\
      Maryland    & 40    & 49                                             & 9                                             & 47      & 7        \\
      Sycamore    & 53    & 65                                             & 12                                            & 62      & 9        \\
      Zuchongzhi2 & 66    & 81                                             & 15                                            & 78      & 12       \\
      IonQ        & 79    & 97                                             & 18                                            & 93      & 14       \\
      Aspen-M     & 80    & 98                                             & 18                                            & 94      & 14       \\
      Eagle       & 127   & 156                                            & 29                                            & 150     & 23       \\ \bottomrule
    \end{tabular}
    \caption{}
    \label{DTABLE}
  \end{subtable}
  \caption{(a) The value of $1- \frac 12\left(\frac 1{\alpha}+\frac 1{1-\beta}-\frac {1}{\alpha(1-\beta)}\right)$. (b) \cmt{Number of saved qubits if applied to several quantum computers .}The largest problem instance that can be handled if applied to several quantum computers. In this table, $n_0$ column stands for the actual number of available qubits, $\max n$ stands for the maximum problem size that can be solved with this device using our framework ensuring $\gamma > \beta$, and $\Delta$ indicates the difference between $n$ and $n_0$. }
  \label{VDTB}
\end{table}

We observe  the right-hand side of \eqref{FINALBOUND} is considerably small, see \tref{VTABLE}. For example, if the classical approximation ratio $\beta=0.89$ and the QAOA approximation ratio $\alpha=0.95$, then $1- \frac 12\left(\frac 1{\alpha}+\frac 1{1-\beta}-\frac {1}{\alpha(1-\beta)}\right)=0.71$ and ${n_0}/{n}$ approaches $85\%$ as $n$ scales up, which means we can save up $15\%$ qubits.  
We also list a table that shows the corresponding largest problem instance that can be solved on current quantum computers.

\paragraph{The trivial case: $\beta=\frac 12$} There is a trivial approach that takes the spirit of the naive algorithm mentioned in \sref{naivealgo}, which achieves $\frac 12$-approximation on $\best$ problem. When a color of $V_0$ is fixed to $s_0$, we can just uniformly and independently color each vertex in $V_1$. Now we prove that this gives $\frac 12$-approximation ratio on \best  \ problem.
\begin{proof}
  For each edge $(x_0,y_1) \in E-E_0-E_1$, it connects two vertices from two partitions $x_0\in V_0, y_1\in V_1$. The $x_0$ is already colored, and there will be $\frac 12$ probability that $y_1$ is colored differently from $x_0$, making $(x_0,y_1)$ becomes a cut with $\frac 12$ probability.

  For each edge $(x_1,y_1) \in E_1$, this edge have $4$ different colorings, and in $2$ out of them $x_1$ and $y_1$ have different colors. Thus $(x_1,y_1)$ becomes a cut with $\frac 12$ probability.

  The expected cut size given by this method is then
  \begin{equation}
    \frac 12|E-E_0|     \ge \frac 12 \best(s_0), \quad \forall s_0.
  \end{equation}

\end{proof}
This $\frac 12$ approximation is low, thus having the overall algorithm do better than $\beta$, as mentioned in \lmref{LEMMA2}, is not a decent goal. So we revisit the approximation ratio $\gamma$ given by the overall algorithm, as indicated by \eqref{GAMMAB} and take $\beta=\frac 12$,
\begin{align}
  \gamma \ge & \alpha\cdot\left(1+(\beta-1)\frac {\max_s\{b_s\}}{\max_s\{a_s+b_s\}}\right)     \\
  \ge        & \alpha\cdot\left(1-\frac {\max_s\{b_s\}}{2\cdot \max_s\{a_s+b_s\}}\right)       \\
  \ge        & \alpha\cdot\left(1-\frac {|E-E_0|}{|E|}\right) = \alpha\cdot \frac {|E_0|}{|E|} \\
  \ge        & \alpha\cdot \frac {n_0(n_0-1)}{n(n-1)},
\end{align}
giving us a relationship between $\gamma$, $\alpha$ and the ratio $\left(n_0(n_0-1)\right)/\left(n(n-1)\right)$. We see that if we want an overall $\gamma_0$ approximation, and given the approximation ratio for QAOA being $\alpha$, it is required to have
\begin{equation}
  \frac {n_0(n_0-1)}{n(n-1)} \ge \frac {\gamma_0}{\alpha}.
\end{equation}

\begin{table}[t]
  \footnotesize
  \centering
  \begin{tabular}{l|l|ll}
    \toprule
    Device                     & $n_0$ & max $n$ & $\Delta$ \\
    \midrule
    Tianmu-1                   & 36    & 37      & 1        \\
    Maryland                   & 40    & 42      & 2        \\
    Sycamore                   & 53    & 55      & 2        \\
    Zuchongzhi2                & 66    & 69      & 3        \\
    IonQ                       & 79    & 83      & 4        \\
    Aspen-M                    & 80    & 84      & 4        \\
    Eagle                      & 127   & 133     & 6        \\
    \textit{delusional device} & 1024  & 1079    & 55       \\ \bottomrule
  \end{tabular}
  \caption{The largest problem instance that can be handled if applied to several quantum computers  when $\alpha=1,\beta=\frac 12$  and $0.9$-approximation is required.}
  \label{AG09}
\end{table}

This is in fact a strict constraint, which is understandable considering our trivial solution to $\best$. To see this, we set $\alpha$ to $1$ which means we assume QAOA obtains accurate solution, then we need $\left(n_0(n_0-1)\right)/\left(n(n-1)\right)\ge \gamma_0$. Even if the goal is just to achieve $0.9$-approximation on Max-Cut (which is not high considering classical computation achieves $0.878$), only a few qubits can be saved, see \tref{AG09}. However, it is worth mentioning that the larger the device is, the more effective this method will be, which can be seen from the last row of the table. So although this trivial $\beta=\frac 12$ case is not practically meaningful at present, it may become useful in the future as the quantum computers scale up.


\section{Experimental results} 
\label{sec:simu}

\begin{figure}[h]
  \centering
  \includegraphics[width=\textwidth]{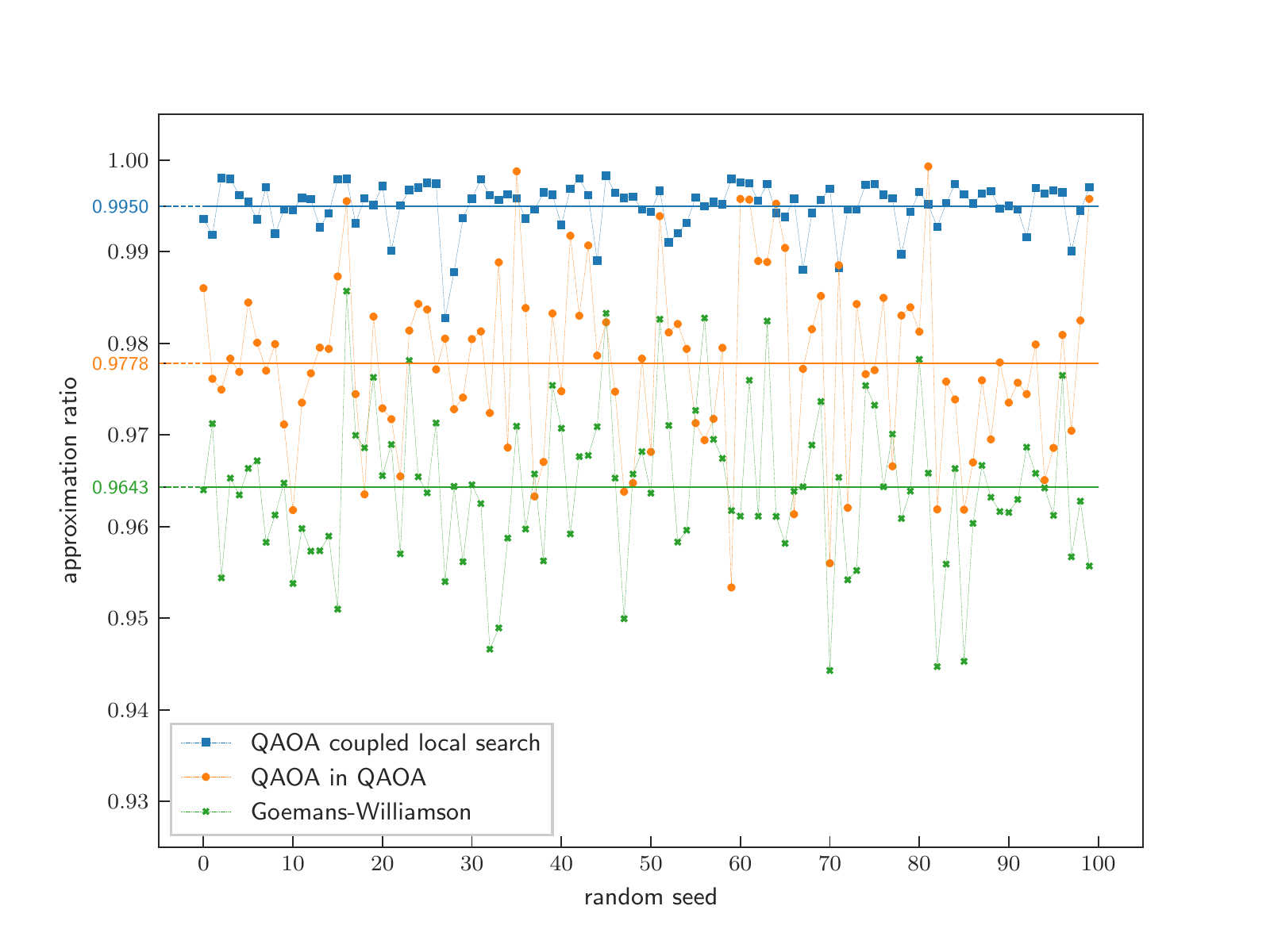}
  \caption{We generate $24$-vertex \ER graphs with $0.8$ connectivity and use $18$-qubit QAOA to solve Max-Cut on them. Blue line (\texttt{QAOA coupled local search}) shows the approximation ratio given by our method, orange line (\texttt{QAOA in QAOA}) by the method in \cite{zhou_qaoa--qaoa_2023} and green line (\texttt{\GW}) by \GW algorithm. We can see from the plot that on most instances our method gives better results. In fact, our method outperforms the other two in $96$ out of the total $100$ instances. The three horizontal lines shows the average approximation ratio given by three methods, and our result is clearly better.}
  \label{SIMRES}
\end{figure}

\cmt{To be modified \tian{In this section, we will discuss the construction of the oracle shown in \eqref{CLSORAC}.} This section composes of two parts. As the previous section assumes the access to an oracle \eqref{CLSORAC} based on classical algorithm, we now discuss how such an oracle might be built. We emphasize that we do not have a full solution \tian{another way???}, but we present a conjecture that leads our way of finding one. The conjecture is discussed and numerically supported in this section. For the second part, we assume the conjecture is true and performed some numerical experiments of our method. Results are compared with previous works and we can observe that our method achieves better approximation in experiments.
}

This section composes of two parts. As the previous section assumes the access to an oracle \eqref{CLSORAC} based on classical algorithm, we now discuss how such an oracle might be built. Building such a oracle proves to be hard as we only allow $o(n)$ space overhead (which is discussed in details in the following text), but we point out a possible direction which is supported by numerical results. For the second part, we propose a heuristic approach for the classical algorithm and perform some numerical experiments on $24$-vertex \ER graph using $18$ qubits. Our method is compared with \GW algorithm and previous work \cite{zhou_qaoa--qaoa_2023}, and the results are shown in \fref{SIMRES}. We can see that method with $18$ qubits is still able to outperform \GW algorithm by a lot, and the performance of the algorithm in \cite{zhou_qaoa--qaoa_2023} varies severely, but most of the time, the performance of our method is better than it. \cmt{\tian{some detailed examples}}

\subsection{On implementing the classical algorithm} 
\label{sec:scsc}

Review the definition of $\best$ (given in \eqref{OMGDEF})
$$
  \best(s_0)    = \max_{s_1 \in \{0,1\}^{n_1}}\cxiss.
$$
It is worth mentioning that the problem above is an extended version of Max-Cut. If we choose $V_0=\emptyset$, then $\cross(s_0,s_1)=0$ and $\best$ will degenerate to Max-Cut problem on $G_1$. Thus, $\best$ problem is as hard as Max-Cut problem, and in \hyperref[BESTGW]{Appendix C}, we demonstrate a polynomial space approximation algorithm that solves $\best$ problem with $.878$ approximation ratio, showing $\best$ and Max-Cut indeed have similar difficulty. The algorithm described in \hyperref[BESTGW]{Appendix C} naturally gives rise to a quantum circuit with $O(n)$ auxiliary qubits -- which is not tolerable as our goal is to save qubits. In fact, due to  \cite{kapralov_optimal_2019}, we know no classical algorithm can break the trivial $\frac 12$ approximation ratio for Max-Cut with $o(n)$ space. Thus, there is no suitable classical algorithm solving $\best$ that can be straightly implemented on a quantum circuit to achieve \eqref{CLSORAC}. However, we only care about the number of qubits needed in our framework, which means we can use arbitrary (polynomial) classical space to prepare for the quantum circuit, which circumvents the constraint from \cite{kapralov_optimal_2019}.


{
In solving $\best$, we can do brute force, i.e., enumerating all possible $s_1$ and take the maximum as defined in \eqref{OMGDEF}. The brute force method is inefficient, and we want to do pruning  to reduce the search space. It is guessed that some $s_1$ is unnecessary to be tried. For example, if $\inner{1}{s_1}$ is too small, then $s_1$ is unlikely to be selected as the optimal choice for any $s_0$. Due to the discussion above, we give the following definition:

\begin{npenv}[$\beta$-guarantee set]
  For arbitrary fixed approximation ratio requirement $\beta$, there is a set $T \subset \zo{n_1}$ such that 
  \begin{equation}
    \forall s_0 \in \zo{n_0}, \frac {\max\limits_{s_1\in T}\cxiss }{\best(s_0)}\ge \beta.
    \label{REQ}
  \end{equation}
  We call such set $T$ a $\beta$-guarantee set. 
\end{npenv}

We ask, what is the minimum number of  $s_1$ elements that a $\beta$-guarantee set should contain? This problem is important because a small $\beta$-guarantee set straightly leads to an efficient optimization workflow (described in \autoref{qaoacadsec}). In the following passage, we will try to find small $\beta$-guarantee sets for random graphs to see what it looks like, and numerical results seem to show that there exist considerably small $\beta$-guarantee sets. Now we describe our method.\cmt{Thus we further discuss \eqref{REQ} and finally lead to our method to verify the conjecture.\tian{the point of the discussion is to explain the proof is hard???}} Consider the process of constructing a $\beta$-guarantee set $T$. For each $s_1\in \zo{n_1}$, if it is added to a set $T$, it will be able to ensure that all elements in a certain set $\cover_{\beta}(s_1)\subset \zo{n_0}$ meet the approximation requirement, namely,
\begin{equation}
  \cover_{\beta}(s_1):=\left\{s_0\mdc \frac {\cxiss}{\best(s_0)}\ge \beta \right\}.
\end{equation}
Thus, our goal of constructing a $\beta$-guarantee set can be reformulated as the well-known the set-covering problem.

\begin{npenv}[Reformulation as set-covering]
  Given $2^{n_1}$ subsets $\{\cover_{\beta}(s_1)\}_{s_1\in \zo{n_1}}$ of $U=\zo{n_0}$, we'd like to construct a set $T \subset \zo{n_1}$ such that
  \begin{equation}
    \bigcup_{s_1\in T}\cover_{\beta}{(s_1)}=U.
  \end{equation}
\end{npenv}

Set-covering is known to be NP-hard \cite{korte_combinatorial_2012}. We make use of a greedy algorithm that achieves $\ln |U|$ approximation ratio, defined in \aref{LOGNCOVER}. This means if we apply it to our $\beta$-guarantee set searching problem, it will be able to output a set $\GS$ that guarantees
\begin{equation}
  \frac {|\GS|}{|\text{minimum } T|}\le \ln 2^{n_1}\in O(n_1).
\end{equation}


\begin{algorithm}
  \caption{$\ln |U|$-approximation set covering}
  \label{LOGNCOVER}
  \KwData {$m$ subsets $\mathcal{C}=\{S\mid S\subset U\}$ that guarantee $\bigcup_{S\in \mathcal{C}}S=U$.}
  \KwResult {Covering set $\mathcal{T}\subset \mathcal{C}$ that satisfies $\bigcup_{T\in \mathcal{T}}T=U$.}
  $\mathcal{T} \gets \emptyset$;

  \While{$\bigcup_{T\in \mathcal{T}}T\not=U$}{
    $S^*\gets \arg \max_{S\in \mathcal{C}}\#\left\{S-\bigcup_{T\in \mathcal{T}}T\right\}$;

    $\mathcal{T}\gets \mathcal{T}\cup \left\{S^*\right\}$\Comment*[r]{Choose the most profitable subset}
  }
  \Return $\mathcal{T}$;
\end{algorithm}

It is impossible to run the full algorithm for relatively large graphs since the time complexity is $O(|\GS|\cdot |\mathcal{C}|\cdot |U|)$, which becomes $O(|\GS|\cdot 2^n)$ in our case. As it is hard to run the algorithm for $U=\{0,1\}^{n_0}$, naturally we want to try the algorithm on a random subset $U_0\subset U$. In other words, instead of finding $\GS$ that satisfies \eqref{REQ}, we randomly sample a $U_0\subset \zo{n_0}$, and use \aref{LOGNCOVER} to calculate $\GS(U_0)$ that guarantees
\begin{equation}
  \forall s_0\in U_0,\quad \frac {\max\limits_{s_1\in \GS(U_0)}\cxiss }{\best(s_0)}\ge \beta.
\end{equation}
For a fixed graph $G$ and partition $(V_0,V_1)$, the more elements we add to $U_0$, the larger $\GS(U_0)$ given by \aref{LOGNCOVER} will be. Moreover, it is clear that for any graph, when $|U_0|=2^0$, $|\GS(U_0)|$ is trivially $1$. Thus, if we  figure out how $|\GS(U_0)|$ increases as $U_0$ grows, we can estimate how large $\GS$ is likely to be. And this is possible for numerical experiments with limited computational power.

In our numerical experiments, we generate $n$-vertex \ER graph with $0.8$ connectivity, and select a dense subgraph with $n_0$ vertices using the heuristic algorithm in \cite{noauthor_heuristic_1994}. Under such settings, we compute how $|\GS({U_0})|$ increases as $|U_0|$ gradually grows to $2^{n_0}$. We first generate a random permutation $p$ of $\{0,1,2,\dots, 2^{n_0}-1\}$, and for a random graph, we compute the $\GS({U_0})$ when $U_0=\{p_0\}, \{p_0,p_1\}, \{p_0,p_1,p_2,p_3\}, \dots ,$ $\{p_0,p_1,\dots, p_{2^{n_0}-1}\}$. The procedure above is done for $20$ random permutations and $20$ random graphs, and the results for $n=35, n_0=27$ are shown in \fref{P35ALL}. It can be observed that $|\GS({U_0})|$ seems to grow almost linearly with $\log |U_0|$, \cmt{\textcolor{blue}{need a linear growth as the comparison in the figure}}indicating the final $|\GS|$ will not be very large.

\begin{figure}[ht]
  \centering
  \begin{subfigure}[b]{0.8\textwidth}
    \centering\includegraphics[width=\textwidth]{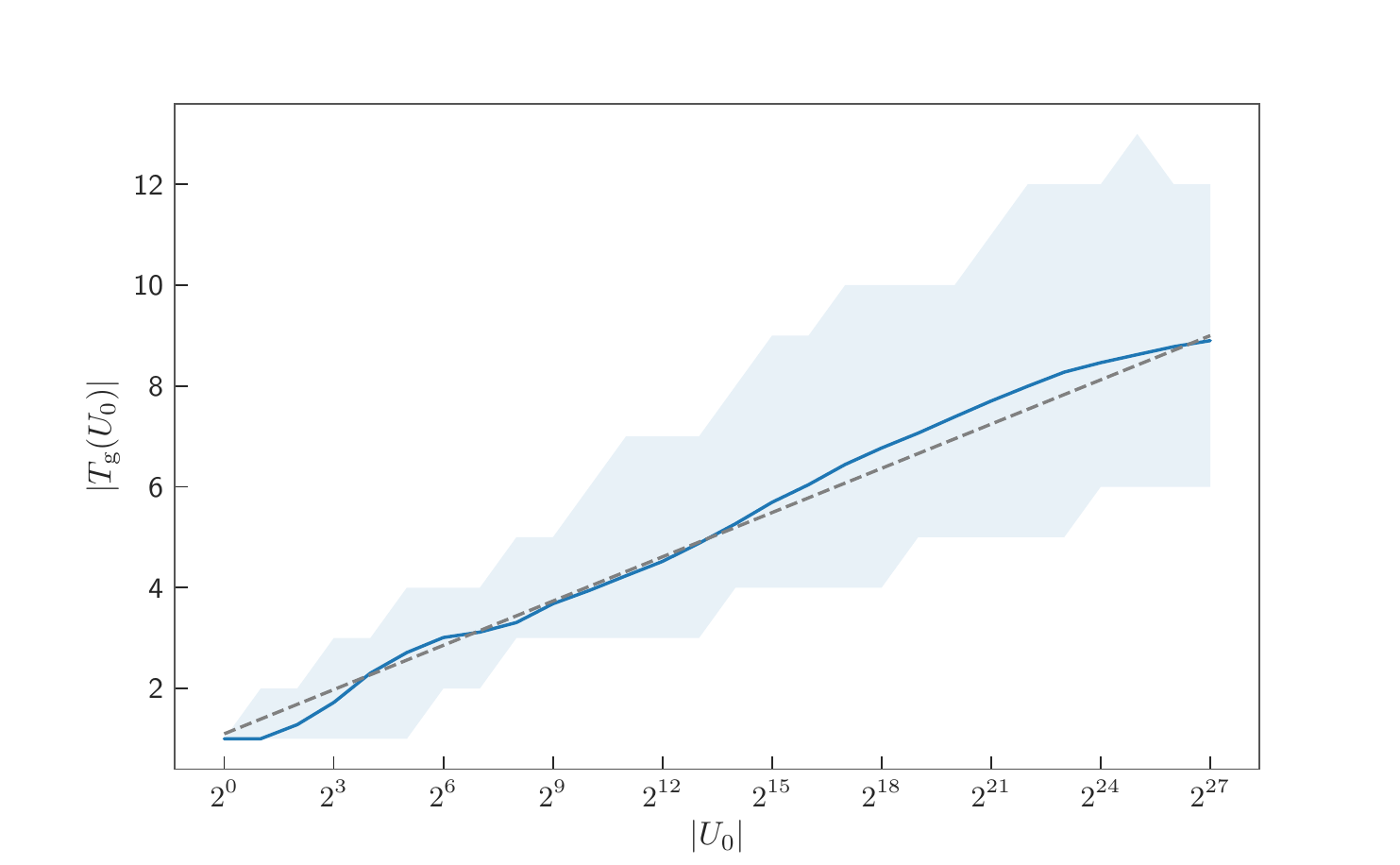}
    \caption{}
    \label{P35ALL}
  \end{subfigure}

  \begin{subfigure}[b]{\textwidth}
    \centering
    \includegraphics[width=0.4\textwidth]{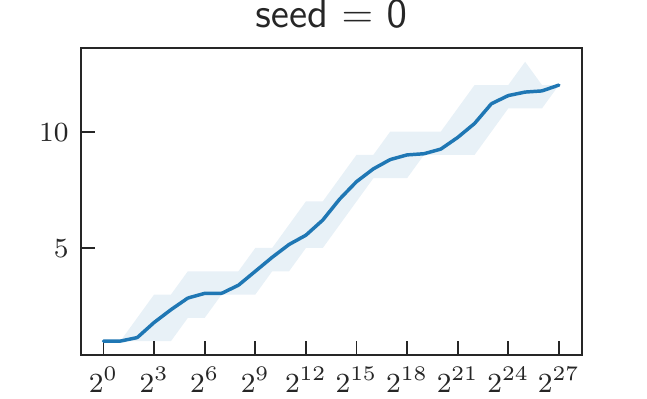}
    \includegraphics[width=0.4\textwidth]{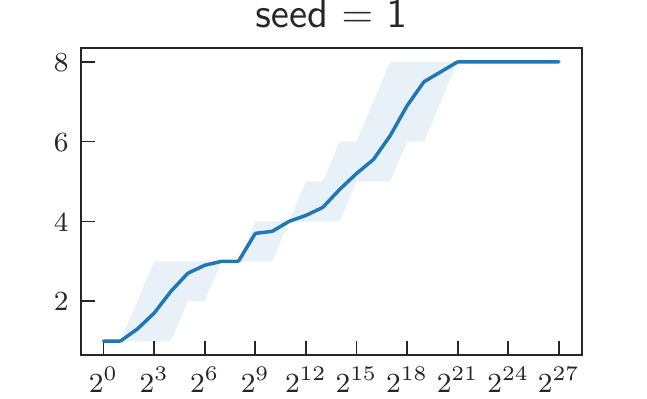}
    \caption{}
    \label{CONCENT}
  \end{subfigure}

  \caption{(a) The increment of $|\GS({U_0})|$ with $n=35,n_0=28,\beta=0.9$, averaged over $20$ random \ER graphs and $20$ random permutations. A straight line (gray) is added for reference. (b) The concentration of $|\GS({U_0})|$ for fixed graphs and $|U_0|$, respectively generated using graphs whose random seeds are shown in their titles. The solid line indicates the average size of $|\GS({U_0})|$, while the shaded area shows the minimum and the maximum of $|\GS({U_0})|$ over all samples and graphs.}
\end{figure}


Moreover, we also observe, that for a fixed graph, when we apply \aref{LOGNCOVER} to it with different $U_0$ of the same size, the sizes of $\GS({U_0})$ display signs of concentration, as seen in \fref{CONCENT}. This concentration suggests that when we add elements into $U_0$ one by one, $|\GS({U_0})|$ grows following an almost identical speed to a final $|\GS|$, regardless of the order of elements added.
Further experiments with different $n$ and $n_0$ are described in \hyperref[SUPEXSEC]{Appendix A}. 

Our exploration shows a tendency that $|\GS|$ seems to grow slowly and will not become very large. But currently we only have numerical results and cannot derive an asymptotic order for $|\GS|$, which, along with how to efficiently find a small $\beta$-guarantee set, remains for future work to tackle.
}

\subsubsection{QAOA circuit based on $\beta$-guarantee set}

\label{qaoacadsec}

\begin{figure}
  \begin{subfigure}[b]{\textwidth}
    \includegraphics[width=\textwidth]{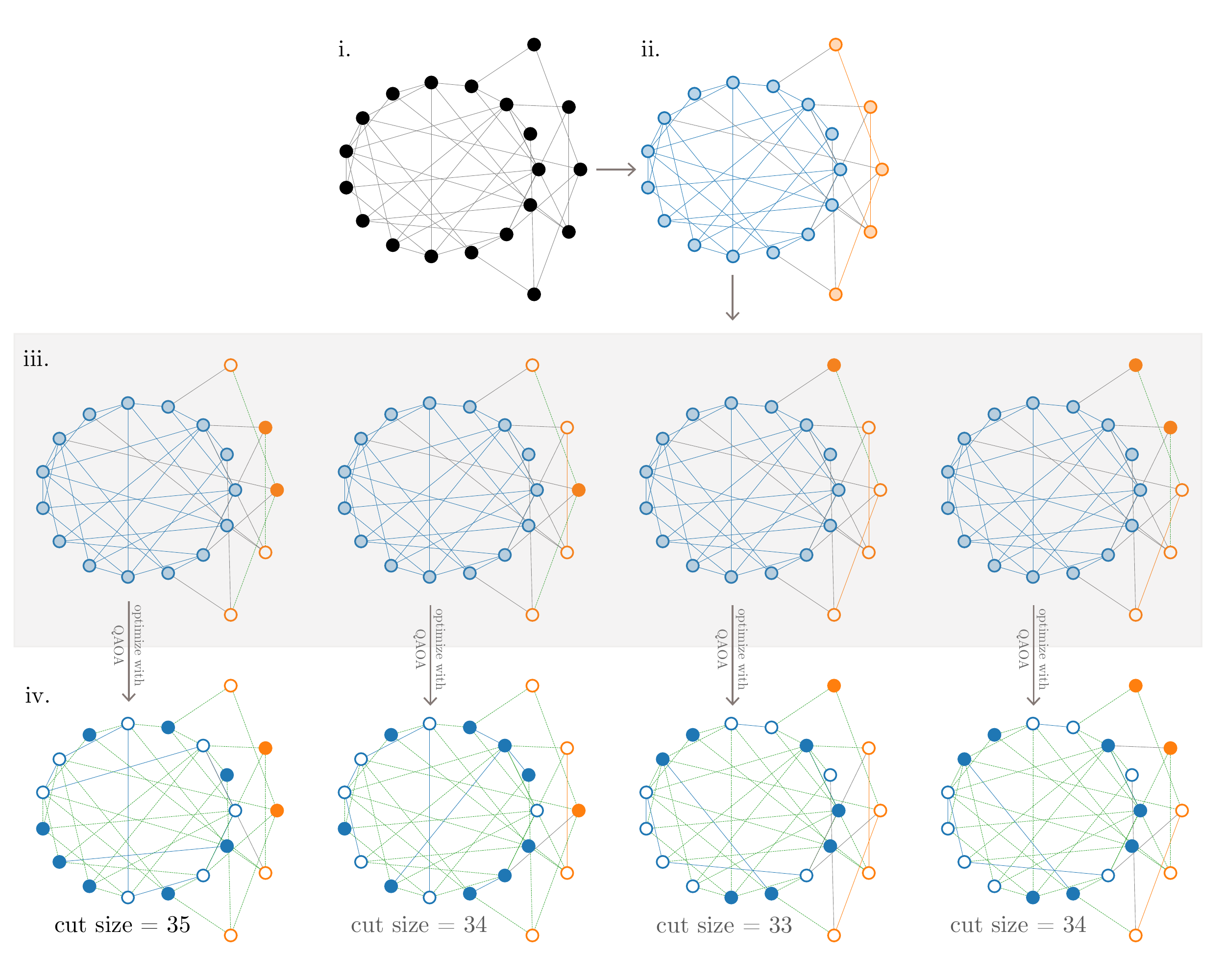}
    \caption{}
    \label{COMBFIG}
  \end{subfigure}
  \begin{subfigure}[b]{\textwidth}
    \centering
    \def\qwb{\qwbundle[alternate]{}}
    \def\myldots{\ \ldots\ \qw}
    \def\myvdots{\ \vdots\ }
    \adjustbox{scale=0.8,center}{
      \begin{quantikz}
        \lstick{} & \gate{R_z}\gategroup[wires=5,steps
          =1,style=dashed]{$\me{-\mi \gamma_k \cdot \cxiss }$} & \qw &\ctrl{2}\gategroup[wires=5,steps
          =7,style=dashed]{$\me{-\mi\gamma_k\cdot \inner{0}{s_0}}$} & \qw &\ctrl{2} & \qw & \qw & \qw &  \myldots & \rstick{}\qw \\
        \lstick{} & \gate{R_z} & \qw&\qw &\qw & \qw & \ctrl{3} & \qw & \ctrl{3} & \myldots & \rstick{}\qw  \\
        \lstick{} & \gate{R_z} & \qw &\targ{} & \gate{R_z} & \targ{} & \qw & \qw & \qw & \myldots & \rstick{}\qw \\
        \lstick{} & \myvdots \\
        \lstick{} & \gate{R_z} & \qw &\qw & \qw & \qw &\targ{} & \gate{R_z} &\targ{} & \myldots & \rstick{}\qw
      \end{quantikz}
    }
    \caption{}
    \label{TQAOASINGLECIRC}
  \end{subfigure}
  \begin{subfigure}[b]{\textwidth}
    \centering
    \adjustbox{scale=0.8,center}{
      \begin{tikzcd}
        \lstick{$x$-th qubit\quad $\ket{c}$} & \gate{R_z\left(-\mi \gamma_k\cdot  \left(\rhox{1}-\rhox{0}\right)\right)}
        &\rstick{$\underbrace{\me{\mi \gamma_k\cdot \left(\rhox{0}+\rhox{1}\right)/2}}_{\text{a global phase factor}}\cdot\me{-\mi \gamma_k\cdot \rhox{c}}\ket{c}$}\qw
      \end{tikzcd}
    }
    \caption{}
    \label{CXIRZ}
  \end{subfigure}
  \caption{(a) Overall workflow. From the original $20$-vertex graph shown in \textit{i}, we select a $15$-vertex dense graph, which is colored blue in \textit{ii}. There exists a $.878$-guarantee set with $4$ elements, and is represented with unfilled and filled orange vertices in \textit{iii}. Based on these $4$ elements, we can generate $4$ QAOA circuits that look like (b), and running optimization separately gives us $4$ possible solutions in \textit{iv}, from which we take the first solution, being the optimum, as our final solution. (b) The circuit for $\ket{s_0}\mapsto \me{-\mi\gamma_k\cdot \left(\inner{0}{s_0}+\cxiss \right)}\ket{s_0}$. (c) Adding a phase factor $\me{-\mi \gamma_k \cdot \rhox{c}}$ (up to a global phase factor). }
\end{figure}

With a $\beta$-guarantee set $T$, the approximate solution to $\best$ will hence be 
\begin{equation}
  \sbest(s_0)=\max_{s_1\in T}\cxiss .
\end{equation}
Assuming QAOA achieves approximation ratio $\alpha$, then we are promised to obtain a cut solution with size of at least
\begin{align}
    & \alpha\cdot \left(\max_{s_0\in\{0,1\}^{n_0}}\left\{\inner{0}{s_0}+\sbest(s_0)\right\}\right)                              \\
  = & \alpha\cdot \left(\max_{s_0\in\{0,1\}^{n_0}}\left\{\inner{0}{s_0}+\max_{s_1\in T}\cxiss \right\}\right)                   \\
  = & \max_{s_1 \in T}\left\{\alpha\cdot \max_{s_0\in\{0,1\}^{n_0}}\left\{\inner{0}{s_0}+\cxiss \right\}\right\}\label{RHSAPP}.
\end{align}
The equation above suggests we can fix $s_1$ and solve
\begin{align}
    & \max_{s_0\in\{0,1\}^{n_0}}\left\{\inner{0}{s_0}+\cxiss \right\}                         \\
  = & \inner{1}{s_1}+\max_{s_0\in\{0,1\}^{n_0}}\left\{\inner{0}{s_0}+\cross(s_0,s_1)\right\},
\end{align}
and take the maximum output over all $s_1$. And the optimization
\begin{equation}
  \max_{s_0\in\{0,1\}^{n_0}}\left\{\inner{0}{s_0}+\cross(s_0,s_1)\right\}\label{TPROBLEM}
\end{equation}
can also be tackled with QAOA, which gets at least $\alpha$ of the optimum. Thus, we can use QAOA to approximately solve \eqref{TPROBLEM} for all $s_1\in T$, and take the maximum answer. The procedure described above achieves solutions no worse than \eqref{RHSAPP}, meaning that this multiple QAOA procedure shares the same performance guarantee with the previous procedure using one single QAOA. The overall procedure is demonstrated using an example in \fref{COMBFIG}.

Now it remains to show how to solve \eqref{TPROBLEM} with QAOA. That is, we will need to implement
\begin{equation}
  \ket{s_0}\mapsto \me{-\mi\gamma_k\cdot (\inner{0}{s_0}+\cross(s_0,s_1))}\ket{s_0}.\label{TQAOASINGLE}
\end{equation}
The first part is easy, and it can be done with \fref{UXY} as we have discussed before. For the second part, we use $s_0(x)$ to denote the color of vertex $x$ in $s_0$, and $s_1(y)$ for the color of vertex $y$ in $s_1$, then
\begin{equation}
  \begin{aligned}
    \cross(s_0,s_1)= & \sum_{(x,y)\in E-E_0-E_1} [s_0(x)\not= s_1(y)] \\
    =                & \sum_{x\in V_0}\sum_{\substack{y\in V_1 \wedge \\ (x,y)\in E-E_0-E_1}}[ s_1(y)=\neg s_0(x)].
  \end{aligned}
\end{equation}
For convenience, let $\Gamma_{x,0/1}$ be the contribution from vertex $x\in V_0$ to cut when it is colored $0/1$, that is
\begin{equation}
  \Gamma_{x,0/1}=\sum_{\substack{y\in V_1 \wedge \\ (x,y)\in E-E_0-E_1}}[ s_1(y)= 1/0],
\end{equation}
then
\begin{equation}
  \cross(s_0,s_1)= \sum_{x\in V_0}\Gamma_{x,s_0(x)}.
\end{equation}
So, for the $x$-th qubit with state $0/1$, we need to add a phase factor $\me{-\mi \gamma_k\cdot \Gamma_{x,0/1}}$. This can be done with $R_z$ gate, as shown in \fref{CXIRZ}. Therefore, the overall circuit for \eqref{TQAOASINGLE} as a part of the QAOA for solving \eqref{TPROBLEM} would look like \fref{TQAOASINGLECIRC}. This circuit consists of only CNOT and $R_z$ gates, therefore remains NISQ-friendly.

\cmt\added{For a $\beta$-guarantee set $T$, the method described in this section requires $|T|$ runs of QAOA optimization. Hopefully that for graphs with certain properties, we can theoretically derive that there exists $T$ such that $|T| \in \poly(n)$, but even if no theoretical bound can be estimated, a practically acceptable size will still be enough. So it will be of significant interest to design an practical algorithm that constructs a
  moderate-size $T$, with or without a theoretical bound for its size.}

\subsection{A heuristic approach based on local search} 
In order to demonstrate the viability of \cmt\tian{our} coupling framework for QAOA circuit design, we set up experiments to show its performance when it is applied to real Max-Cut instances. As pointed out in previous sections, currently we are not able to efficiently construct $\beta$-guarantee set as required, so a heuristic method needs to be devised.

Here we give a heuristic method that converges in polynomial time and demonstrate its numerical results, compared with previous methods. In previous sections, we regard the classical algorithm that finds good coloring for $V_1$ as a subprocedure called by QAOA which together forms our coupling framework. But now what we actually do is to first fix a coloring $s_1$ of $V_1$, and compute the optimal coloring of $V_0$ corresponding to it. Thus, we can \cmt\tian{iteratively} update the chosen $s_1$ following a local search scheme to gradually reach a good result. The overall structure of our heuristic method will look like  \aref{QAOALS}.

\begin{algorithm}[h]
  \caption{Coupling QAOA with local search}
  \label{QAOALS}
  \KwData {A graph $G=(V,E)$.}
  Find a dense subgraph $G_0=(V_0,E_0)$ for $G$;

  Setup an initial coloring for $V_1$, denoted by $s_1$;

  $p\gets \mathbb{E}\left(\text{QAOA}\big| V_1 \text{ is colored } s_1\right)$;

  \While{true}{

    \For{$u\in V_1$}{
      $s_1'\gets s_1$ with the color of $u$ fliped;

      $q\gets \mathbb{E}\left(\text{QAOA}\big| V_1 \text{ is colored } s_1'\right)$;

      $r_u\gets (q,s_1')$;
    }
    $(q,s_1')\gets \max r$;

    \eIf{$q>p$}{
      $(p,s_1)\gets (q,s_1')$;
    }{
      break;
    }
  }
  \Return $p$;
\end{algorithm}



We also tested the algorithm proposed by \cite{zhou_qaoa--qaoa_2023} for comparison. We strengthen that \cite{zhou_qaoa--qaoa_2023} is able to split the graph arbitrarily so it is able to solve larger Max-Cut instances using small-scale quantum computers , but here we only consider such occasions that the number of qubits is slightly smaller than the graph size, and we demonstrate that our method has the potential to make better use of such amount of available qubits. Experiment results are shown in \fref{SIMRES}.

\section{Conclusion}
\label{sec:conc}
In this work, we propose a coupling framework for QAOA circuit design and demonstrate it using Max-Cut as an example. Our method replaces a certain portion of quantum computation with classical computation to save quantum resources at the expense of losing a bit of approximation performance.

Our method relies on an efficient quantum-implementable classical algorithm that approximately solves the $\best$ problem defined in \eqref{OMGDEF}. As the problem is as hard as Max-Cut and the classical algorithm can only consume $o(n)$ qubits when implemented as a quantum oracle, currently we are not able to design a classical algorithm as required\cmt{\tian{hardness???}}, but we present some numerical observation that will lead our way towards one. However, despite designing a required algorithm with theoretic guarantee being hard, we give a heuristic method inspired by local search algorithm. Numerical experiments show good performance which suggests our coupling framework might be able to help give more satisfactory solutions to larger Max-Cut problem instances using currently limited number of qubits. Future work remains to design a classical algorithm that solves \eqref{OMGDEF} with rigorous analysis and performance guarantee.

\section{Data and code availability}
The datasets and the source code that generates them are available at \cite{lu_lucidaluqaoa--fewer-qubits_2023}.

\FloatBarrier

\begingroup
\sloppy
\printbibliography
\endgroup

\appendix

\section{Supplementary experiments for $\GS$ size}
\label{SUPEXSEC}

We change the size of graph, approximation ratio requirement $\beta$, and split ratio ${n_0}/{n}$ and compute $\GS(U_0)$. Experiments under different setups all give positive results, \cmt\added{showing slow (non-exponential) increment of $\log |\GS(U_0)|$ in accordance with $|U_0|$,} see \fref{SUPEX}.

\begin{figure}[h]
  \centering
  \begin{subfigure}[b]{0.49\textwidth}
    \includegraphics[width=\textwidth]{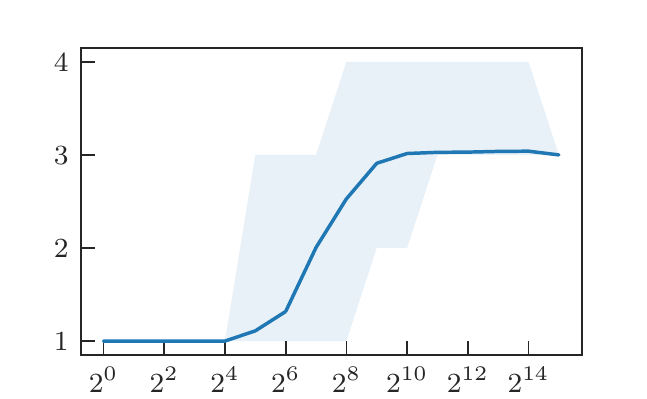}
    \caption{$n=32, {n_0}/n=0.5, \beta=0.9$}
  \end{subfigure}
  \begin{subfigure}[b]{0.49\textwidth}
    \includegraphics[width=\textwidth]{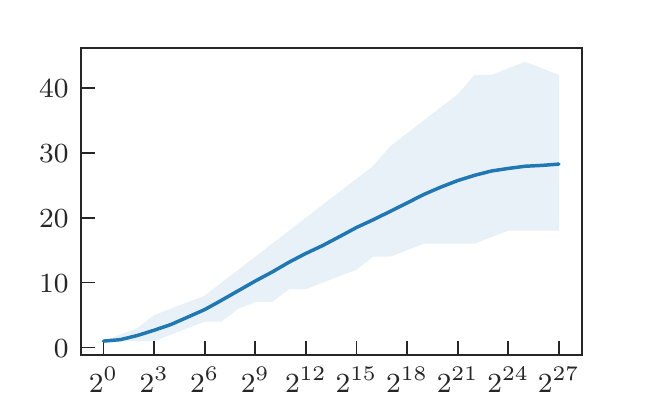}
    \caption{$n=35, {n_0}/n=0.8, \beta=0.95$}
  \end{subfigure}

  \begin{subfigure}[b]{0.49\textwidth}
    \includegraphics[width=\textwidth]{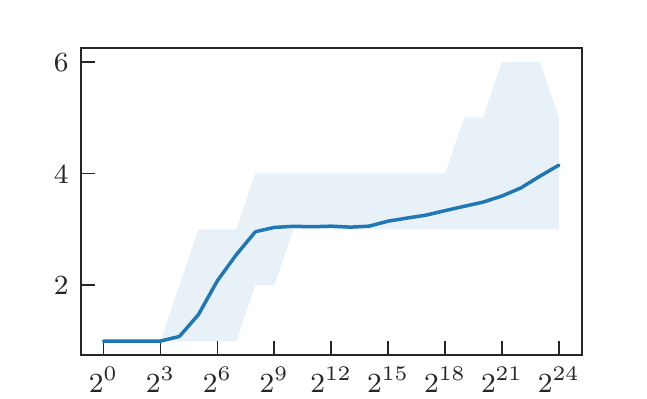}
    \caption{$n=40, {n_0}/n=0.625, \beta=0.9$}
  \end{subfigure}
  \begin{subfigure}[b]{0.49\textwidth}
    \includegraphics[width=\textwidth]{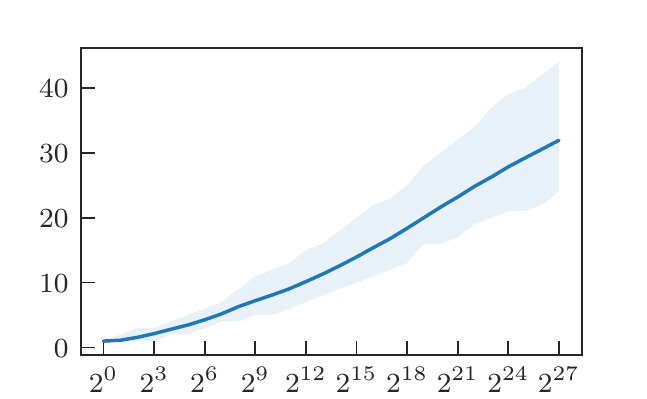}
    \caption{$n=50, {n_0}/n=0.8, \beta=0.95$}
  \end{subfigure}

  \caption{Supplementary experiments for $\GS$ size.}
  \label{SUPEX}
\end{figure}

\section{Experiment settings}
QAOA circuit for Max-Cut only involves $\exp\left(\mi\theta X\right)$ and $\exp\left(\mi\theta ZZ\right)$ gates, which is simply element-wise vector multiplication instead of matrix-vector multiplication under $\{\ket{+},\ket{-}\}^{\otimes n}$ and $\{\ket{0},\ket{1}\}^{\otimes n}$, respectively, while the basis transform between the two basis can be done with fast Walsh-Hadamard transform \cite{yarlagadda_hadamard_1996}. So here we implement a statevector simulator dedicated to QAOA with the help of FWHT, and we use \textit{Numba} \cite{lam_numba_2015} library to achieve parallel speedup with GPU.

In order to fully present the power of our circuit design and prevent performance loss due to unsatisfactory optimization results given by QAOA, we want to have QAOA approximation ratio as close to $1$ as possible. Therefore, we add much redundancy to $p$, the number of QAOA layers -- $p$ is set to $n$ when we use QAOA to perform optimization over $n$-qubit problem instance. More layers means more parameters which leads to better circuit expressibility, thus we will more likely reach a state with higher expectation.

Like many related works, we use Broyden–Fletcher–Goldfarb–Shanno (BFGS) method \cite{broyden_convergence_1970, fletcher_new_1970, goldfarb_family_1970, shanno_conditioning_1970} to optimize QAOA parameters for its quasi-Newton nature provides quick convergence rate. We choose the implementation of BFGS via \textit{scipy} \cite{virtanen_scipy_2020} library , and limit the maximum number of iterations to $100$, with other parameters of BFGS left default. And to avoid getting stuck in abysmal local optima, we adopt the parameter initialization method given by \cite{sack_quantum_2021} and set its hyper-parameter $\delta t=0.56$.

While implementing our coupling method, we select dense subgraphs using a heuristic algorithm given in \cite{noauthor_heuristic_1994} and partition the graph into two parts $(V_0,V_1)$, and run the algorithm defined in \aref{QAOALS} starting with $S_0=\emptyset$. And when implementing \cite{zhou_qaoa--qaoa_2023}, we use the random partition policy described in the paper and respectively use $|V_0|$-qubit QAOA and $|V_1|$-qubit QAOA to solve Max-Cut inside the two subgraphs, and numerically compute the expected output according to the two resultant probability distributions. When implementing \cite{zhou_qaoa--qaoa_2023}, the parameters are changed to have better optimization performance: we set $p=2n$ and iteration limit of BFGS to $500$ up to try to decrease the performance loss of \cite{zhou_qaoa--qaoa_2023} due to the inaccuracy of QAOA.

When testing the performance of \GW algorithm, we use the SDP solver provided in \textit{MOSEK} \cite{aps_introducing_2023}, wrapped with \textit{cvxpy} \cite{diamond_cvxpy_2016} package.

\section{Polynomial space approximation algorithm solving $\best$ problem}
\label{BESTGW}

When polynomial space (instead of only logarithmic) is allowed, $\best$ problem can be approximately solved using a modified version of \GW algorithm, enjoying the same $.878$ approximation guarantee. But note that this algorithm requires $\Omega(n)$ space, thus it cannot be straightly used in constructing \eqref{CLSORAC}.

We first review the workflow of \GW algorithm. Max-Cut has the following standard QUBO (Quadratic unconstrained binary optimization) formalism
\begin{equation}
  \max \left\{\sum_{(x, y) \in E} \frac{1-c_x\cdot c_y}2 \mdc c_x^2=1\right\}.\label{eq58}
\end{equation}
\GW algorithm uses unit real vectors $v_x\in \mathbb{R}^n$ to approximate $c_x$, relaxing \eqref{eq58} as
\begin{equation}
  \max \left\{\sum_{(x,y) \in E} \frac {1-v_x\cdot v_y}{2}:\left\|v_{x}\right\|^2=1\right\}.
\end{equation}
Consider a matrix $M$ with entries defined as $M_{xy}:=\braket{v_x}{v_y}$, then it is known that $M$ is positive semidefinite. Moreover, for any positive semidefinite $M$, there exists $\{v_x\}$ satisfying $M_{xy}:=\braket{v_x}{v_y}$. Thus, there is a one-to-one correspondence between $\{v_x:\|v_x\|^2=1\}$ and $\{M:M_{xx}=1\}$. Now let $L$ be the laplacian matrix of the graph with entries defined as
\begin{equation}
  L_{xy}:=\left\{
  \begin{aligned}
     & \deg(x), & x=y,              \\
     & -1,      & (x,y)\in E,       \\
     & 0,       & \text{otherwise,}
  \end{aligned}
  \right.
\end{equation}
then,
$$
  \begin{aligned}
    \frac {\mathbf{tr}(LM)}4=\frac 14\sum_{x\in V}\sum_{y\in V}L_{xy}(v_x\cdot v_y)= & \frac 14\left(\sum_{x\in V}\deg(x)-\sum_{(x,y)\in E} 2\cdot v_x\cdot v_y\right) \\
    =                                                                                & \sum_{(x,y)\in E} \frac {1-v_x\cdot v_y}2.
  \end{aligned}
$$
From the discussion above, \eqref{eq58} is relaxed to the following SDP formulation:
\begin{equation}
  \max_{M\succeq 0}\left\{\mathbf{tr}(LM)/4:M_{xx}=1\right\}.\label{eq61}
\end{equation}
From \eqref{eq61}, we can get the optimal $\{v_x\}$, after which we select a random hyperplane through the origin, separating $\{v_x\}$ into two sets, which forms a cut.
\begin{lemma}[{\cite[Theorem~3.3]{goemans_improved_1995}}]
  \label{lemma3}
  For arbitrary set of $\{v_x\}$, the expected size of the cut obtained by separating the set using a random hyperplane, is at least $.878$ times $\mathbf{tr}(LM)/4$.
\end{lemma}
As \eqref{eq61} is a relaxation of Max-Cut \eqref{eq58}, the optimal $\mathbf{tr}(LM)/4$ will be no less than Max-Cut. This fact along with \lmref{lemma3} implies \lmref{lemma4}.
\begin{lemma}
  \label{lemma4}
  \GW algorithm have $.878$ approximation ratio on Max-Cut problem.
\end{lemma}

To approximately solve the modified version of Max-Cut, $\best$, only slight amendment to the described \GW algorithm is required. In $\best$, for some vertices, their colors are antecedently assigned, so  we fix their $\{v_x\}$ vectors in advance. For vertices which are assigned different colors, we need to make sure their corresponding vectors appear at different sides of the random hyperplane, and to do this, we fix the vectors so that if two vertices have different colors, their corresponding vectors point to opposite directions. One possible choice is
\begin{equation}
  v_x:=\begin{cases}
    (+1,0,\dots,0), & \text{color of $x$ is $0$}, \\
    (-1,0,\dots,0), & \text{color of $x$ is $1$}.
  \end{cases}
\end{equation}
And under such configuration, the entries of $M$ subject to further constraints. If $z$ is a free vertex, the vertices whose colors are fixed $0$ are $x_1,x_2,\dots,x_a$ and the vertices whose colors are fixed $1$ are $y_1,y_2,\dots,y_b$, then for each $z$,
\begin{align}
  M_{zx_i}-M_{zx_{i+1}} & =0\quad (\forall 1\le i\le a-1), \\
  M_{zy_i}-M_{zy_{i+1}} & =0\quad (\forall 1\le i\le b-1), \\
  M_{zx_1}+M_{zy_1}     & =0.
\end{align}
And,
\begin{align}
  M_{x_iy_j}=-1\quad (\forall 1\le i\le a,1\le j\le b).
\end{align}
Then we solve \eqref{eq61} under the additional constraints above. Similar to the reasoning towards \lmref{lemma4}, we have
\begin{lemma}
  The modified \GW algorithm described in this section have $.878$ approximation ratio on $\best$ problem.
\end{lemma}
The existence of efficient classical non-trivial approximation algorithm solving $\best$ suggests that this problem, non-rigorously speaking, has close difficulty to original Max-Cut problem.

\end{document}